%% file: revcat-final.tex
\documentclass[submission,copyright,creativecommons]{eptcs}
\providecommand{\event}{MFPS 2021} 
\usepackage{breakurl}             

\RequirePackage[l2tabu, orthodox]{nag}

\usepackage{amsthm,amsmath,amssymb}
\usepackage{mathtools}
\let\saveamalg\amalg
\let\amalg\relax
\usepackage{mathabx}
\let\amalg\saveamalg
\usepackage{savesym}
\usepackage{amsmath}
\usepackage{bm}
\usepackage{cite}
\usepackage[all,cmtip]{xy}
\usepackage{appendix}
\usepackage{hyperref}
\usepackage{stmaryrd}
\usepackage{comment,color}
\usepackage{subfig}
\usepackage{enumerate,multicol}
\usepackage{microtype} 
\usepackage{wrapfig}
\usepackage{url}

\input{notations.tex}

\input{prooftree.tex}

\newtheorem{theorem}{Theorem}[section]
\newcommand{\kindoftheorem}{Theorem}

\newtheorem{proposition}[theorem]{Proposition}
\newtheorem{corollary}[theorem]{Corollary}
\newtheorem{lemma}[theorem]{Lemma}

\theoremstyle{definition}
\newtheorem{definition}[theorem]{Definition}
\newtheorem{assumption}[theorem]{Assumption}

\theoremstyle{remark}

\DeclareRobustCommand{\amalg}{\mathop{\text{\fakecoprod}}}
\newcommand{\fakecoprod}{%
  \sbox0{$\prod$}%
  \smash{\raisebox{\dimexpr.9625\depth-\dp0}{\scalebox{1}[-1]{$\prod$}}}%
  \vphantom{$\prod$}%
}

\usepackage{tikz}
\usetikzlibrary{positioning,arrows,decorations.pathmorphing}
\tikzset{node distance=15mm, auto}
\tikzstyle{every node}=[font=\footnotesize]
\tikzstyle{morphism}=[font=\scriptsize]

\DeclareMathOperator{\Total}{Total}
\DeclareMathOperator{\Split}{Split}
\DeclareMathOperator{\Par}{Par}
\DeclareMathOperator{\Inv}{Inv}
\DeclareMathOperator{\Hom}{Hom}

\newcommand{\ridm}{\overline} 
\newcommand{\ie}{\emph{i.e.}}
\newcommand{\eg}{\emph{e.g.}}

\newcommand{\parcat}[1]{\Par(\overline{\cat{#1}},\widehat{\mathcal{M}_{\text{gap}}})}
\newcommand{\Set}{\cat{Set}}
\newcommand{\bigjoin}{\bigvee}
\newcommand{\dayconv}{\circledast}

\newcommand{\Ess}{\mathcal{S}}

\newcommand{\tot}{\xrightarrow}
\newcommand{\fold}{\mathrm{fold}}
\newcommand{\unfold}{\mathrm{unfold}}
\newcommand{\sem}[1]{\ensuremath{\llbracket #1 \rrbracket}}

\title{Join Inverse Rig Categories for Reversible Functional Programming, and Beyond}

\author{
Robin Kaarsgaard
\institute{University of Edinburgh \\ Edinburgh, United Kingdom}\\
\email{robin.kaarsgaard@ed.ac.uk}
\and
Mathys Rennela
\institute{INRIA\\ Paris, France}\\
\email{mathys.rennela@inria.fr}
}
\def\titlerunning{Join inverse rig categories for reversible functional programming, and beyond}
\def\authorrunning{R. Kaarsgaard \& M. Rennela}

\begin{document}

\maketitle

\begin{abstract}
Reversible computing is a computational paradigm in which computations are deterministic in both the forward and backward direction, so that programs have well-defined forward \emph{and} backward semantics. We investigate the formal semantics of the reversible functional programming language Rfun.  For this purpose we introduce join inverse rig categories, the natural marriage of join inverse categories and rig categories, which we show can be used to model the language Rfun, under reasonable assumptions. These categories turn out to be a particularly natural fit for reversible computing as a whole, as they encompass models for other reversible programming languages, notably Theseus and reversible flowcharts. This suggests that join inverse rig categories really are the categorical models of reversible computing.

\end{abstract}

\section{Introduction}
\label{sec:intro}

Since the early days of theoretical computer science, the quest for the mathematical description of (functional) programming languages has led to a substantial body of work.
In reversible computing, every program is \emph{reversible}, i.e., both \emph{forward} and \emph{backward} deterministic. But why study such a peculiar paradigm of computation at all?
While the daily operations of our computers are irreversible, the physical devices which execute them are fundamentally reversible. In the paradigm of quantum computation, the physical operations performed by a scalable quantum computer intrinsically rely on quantum mechanics, which is reversible. Landauer \cite{landauer} has 
famously argued, through what has later been coined \emph{Landauer's
principle}, that the erasure of a bit of information is inexorably linked to
the dissipation of energy as heat (which has since seen both formal~\cite{abramsky-horsman} and experimental~\cite{berut} verification). On its own, this constitutes a reasonable
argument for the study of reversible computing, as this model of computation sidesteps this otherwise inevitable energy dissipation by avoiding the erasure of information altogether.

And although studying reversibility can be motivated by issues raised by the laws of thermodynamics which arguably constitute a theoretical limit of Moore's law \cite{moore}, reversibility arises not only in quantum computing (see e.g., \cite{altenkirch-grattage}), but also has its own circuit models~\cite{toffoli,fredkin}, Turing machines~\cite{bennett,axelsen} and automata~\cite{kutrib1,kutrib2}. Moreover, the notion of reversibility has seen applications in areas spanning from high-performance computing~\cite{schordan} to process calculi \cite{cristescu-krivine-varacca} and robotics \cite{schultz-bordignon-stoy,schultz-laursen-ellekilde-axelsen}, to name a few.

There is an increasing interest for the theoretical study of the reversible computing paradigm (see, e.g., \cite{aman-foundations}). The present work provides a building block in the study of reversible computing through the lens of category theory.

\subsection{Reversible programming primer} 
\label{sec:reversible_programming_primer}
Almost all programming languages in use today (with some notable exceptions)
guarantee that programs are \emph{forward deterministic} (or simply
\emph{deterministic}), in the sense that any current computation state uniquely
determines the next computation state. Few such languages, however, guarantee
that programs are \emph{backward deterministic}, i.e.~that any current
computation state uniquely determines the \emph{previous} computation state.
For example, assigning a constant value to a variable in an imperative
programming language is forward deterministic, but not backward deterministic
(as one generally has no way of determining the value stored in the variable
prior to this assignment).\\

\tikzset{cross/.style={cross out, draw, 
         minimum size=2*(#1-\pgflinewidth), 
         inner sep=0pt, outer sep=0pt}}
\noindent 
\begin{minipage}{.24\textwidth}
\centering
\textit{Forward non-determinism}\\
\ \\
\begin{tikzpicture}
\node[circle,fill=gray] (current) {};
\node (phantom) [right=10mm of current] {};
\node[circle,fill=black] (next) [above=3mm of phantom]{};
\node[circle,fill=black] (other) [below=3mm of phantom] {};

\draw[->,thick] (current) to node [morphism] {} (next);
\draw[->,thick] (current) to node [morphism] {} (other);
\end{tikzpicture}
\end{minipage}
\hfill
\begin{minipage}{.24\textwidth}
\centering
\textit{Forward\\ determinism}\\
\ \\
\begin{tikzpicture}
\node[circle,fill=gray] (current) {};
\node (phantom) [right=10mm of current] {};
\node[circle,fill=black] (next) [above=3mm of phantom]{};
\node[circle,fill=black] (other) [below=3mm of phantom] {};

\draw[->,thick] (current) to node [] {} (next);
\draw[->,thick] (current) to node [anchor=center, midway, red] {\huge $\mathbf{\bigtimes}$} (other);
\end{tikzpicture}
\end{minipage}
\hfill
\begin{minipage}{.24\textwidth}
\centering
\textit{Backward non-determinism}\\
\ \\
\begin{tikzpicture}
\node (phantom) {};
\node[circle,fill=gray] (current) [right=10mm of phantom] {};
\node[circle,fill=gray!35] (previous) [above=3mm of phantom]{};
\node[circle,thick,fill=gray!35] (other) [below=3mm of phantom] {};

\draw[->,thick] (previous) to node [morphism] {} (current);
\draw[->,thick] (other) to node [] {} (current);
\end{tikzpicture}
\end{minipage}
\hfill
\begin{minipage}{.24\textwidth}
\centering
\textit{Backward determinism}\\
\ \\
\begin{tikzpicture}
\node (phantom) {};
\node[circle,fill=gray] (current) [right=10mm of phantom] {};
\node[circle,fill=gray!35] (previous) [above=3mm of phantom]{};
\node[circle,thick,fill=gray!35] (other) [below=3mm of phantom] {};

\draw[->,thick] (previous) to node [morphism] {} (current);
\draw[->,thick] (other) to node [anchor=center, midway, red] {\huge $\mathbf{\bigtimes}$} (current);
\end{tikzpicture}
\end{minipage}
\ \\

Programming languages which guarantee both forward and backward determinism of
programs are called \emph{reversible}. Though there are numerous examples of
such reversible programming languages (see, \eg, \cite{theseus,janus}), we
focus here on the reversible functional programming language
Rfun~\cite{yokoyama-axelsen-glueck-rc2011}.

\begin{wrapfigure}{r}{0.25\linewidth}
\input{rfun_fib.tex}
\input{rfun_fib_inv.tex}
\end{wrapfigure}

Rfun is an untyped reversible functional programming language (see an example program, computing Fibonacci pairs, to the right, due to \cite{yokoyama-axelsen-glueck-rc2011}) similar in style
to the Lisp family of programming languages. As a consequence of reversibility,
all functions in Rfun are \emph{partial injections}, \ie, whenever some
function $f$ is defined at points $x$ and $y$, $f(x) = f(y)$ implies $x=y$.
Being untyped, values in Rfun come in the form of Lisp-style symbols and
constructors. 

Pattern matching and variable binding is supported by means of
(slightly restricted) forms of case expressions and let bindings, and iteration
by means of general recursion. These restriction ensure, essentially, that the inverse of a case expression is also a case expression.

In order to guarantee backward determinism (and consequently reversibility),
Rfun imposes a few restrictions on function definitions not usually present in
irreversible functional programming languages. Firstly, while a given variable
may only appear once in a pattern (as is also the case in, \eg, Haskell and the
ML family), it must also occur exactly once in the body.
Secondly, the result of a function call must be bound in a let binder before
use. Thirdly, the leaf expression of any branch in a case expression must
not match \emph{any} leaf expression in a branch preceding it: This is known as the \emph{symmetric first-match policy}, and guarantees that case-expressions can be evaluated normally (i.e., in prioritised order from top to bottom) without impacting reversibility.

Together, these
three restrictions guarantee reversibility. One might wonder whether these
hinder expressivity too much; fortunately, this is not so, as Rfun is r-Turing
complete~\cite{yokoyama-axelsen-glueck-rc2011}, that is: Rfun can simulate any
reversible Turing machine \cite{bennett}.

Another peculiarity of Rfun has to do with duplication of values. Even though
values \emph{can} be (de)duplicated reversibly, the linear use policy on
variables hinders this. To allow for (de)duplication of values, a special
\emph{duplication/equality} operator $\lfloor \cdot \rfloor$ is introduced (see Figure~\ref{fig:dupeq}). Note that this operator can be used both as a value and as a pattern; in the
former use case, it is used to (de)duplicate values, and in the latter, to test
whether values are identical.


As a reversible programming language, Rfun has the property that it is
syntactically closed under inverses. That is to say, if $p$ is an Rfun program,
there exists another Rfun program $p'$ such that $p'$ computes the semantic
inverse of $p$. This is witnessed by a \emph{program
inverter}~\cite{yokoyama-axelsen-glueck-rc2011}. For example, the inverses of the
addition and Fibonacci pair functions shown earlier are given above.

\begin{wrapfigure}{r}{0.5\linewidth}
\begin{align*}
  \lfloor \langle x \rangle \rfloor & = \langle x, x \rangle \\
  \lfloor \langle x, y \rangle \rfloor & = \left\{ \begin{array}{l l}
    \langle x \rangle & \text{if } x = y \\
    \langle x, y \rangle & \text{if } x \neq y
  \end{array} \right.
\end{align*}
\caption{The operator $\lfloor \cdot \rfloor$.}
\label{fig:dupeq}
\end{wrapfigure}

\subsection{Join inverse category theory}

In Section~\ref{sec:join-inv}, we focus on categorical structures in which one can conveniently model the mathematical foundations of reversible computing.
A brainchild of Cockett \& Lack~\cite{cockett1,cockett2,cockett3}, restriction categories
are categories with an abstract notion of partiality, associating to each morphism $f : A \to B$ a ``partial identity'' $\ridm{f} : A \to A$ satisfying $f \circ \ridm{f} = f$ and other axioms. 
In particular, this gives rise to \emph{partial isomorphisms}, which are morphisms $f : A\to B$ associated with partial inverses $f^\dagger:B \to A$ such that $f^\dagger \circ f = \ridm{f}$ and $f \circ f^\dagger = \ridm{f^\dagger}$. In this context, a \textit{total} map is a morphism $f: A\to B$ such that $\ridm{f}$ is the identity on $A$.

A restriction category in which all morphisms are partial isomorphisms is called an \emph{inverse category}. 
Such categories have been considered by the semigroup community for decades~\cite{kastl,lawson,hines-thesis}, though they have more recently been rediscovered in the framework of restriction category theory, and considered as models of reversible computing~\cite{axelsen-kaarsgaard-fossacs16,kaarsgaard-axelsen-glueck-jlamp,giles-thesis,guo-thesis}.
The category $\cat{PInj}$ of sets and partial injections is the canonical example of an inverse category. In fact, every (locally small) inverse
category can be faithfully embedded in it~\cite{heunen,cockett1}.

In line with \cite{axelsen-kaarsgaard-fossacs16}, we focus here on a particular class of domain-theoretic inverse categories, namely join inverse categories. Informally, join inverse categories are inverse categories in which joins (i.e. least upper bounds) of morphisms exist in such a way that the partial identity of a join is the join of the partial identities, among with other coherence axioms.

In this setting, join inverse rig categories are join inverse categories equipped with monoidal products $X \otimes Y$ and monoidal sums $X \oplus Y$ for every pair of objects $X$ and $Y$, together with a isomorphisms which distribute products over sums and annihilate products with the additive unit, subject to preservation of joins and the usual coherence laws of rig categories. Every inverse category embeds in such a join inverse rig category via Wagner-Preston (see also Theorem~\ref{thm:bimonoidal_embedding}).

\subsection{Categorical models of reversible computing}

The present work is predominantly concerned with the axiomatization of categorical models of reversible computing, following similar approaches based on presheaf-theoretic models, in the semantics of quantum computing \cite{malherbe-scott-selinger,rennela-staton-furber-qpl16,rennela-staton-mfps31,ewire} and reversible computing \cite{axelsen-kaarsgaard-fossacs16}.

In order to construct denotational models of the terms of the language Rfun, we adopt the categorical formalism of join inverse rig categories. Since morphisms in inverse categories are all partial isomorphisms, inverse categories have been suggested as models of reversible functional programming languages~\cite{giles-thesis}, and the presence of joins has been shown~\cite{axelsen-kaarsgaard-fossacs16,kaarsgaard-axelsen-glueck-jlamp} to induce fixed point operators for modelling reversible recursion.

In short, Rfun is an untyped first-order language in which the arguments of functions are organized in left expressions (patterns) given by the grammar 
\[
l ::= x \mid c(l_1,\ldots,l_n) \mid \lfloor l \rfloor 
\]
where variables $x$ and constructors $c$ are taken from denumerable alphabets $\mathcal{V}$ respectively $S$. In other words, a function definition $f\, l \triangleq l'$ takes a pattern $l$ as argument and organizes its output as an expression $l'$.

Our work involves an unorthodox way of thinking about the denotational semantics of a function in a functional programming language to be explicitly constructed piecemeal by the (partial) denotations of its individual branches.
Categorically, this means that the denumerable alphabet $S$ is denoted by the (least) fix point of the functor $F:X \mapsto X \oplus 1$, which is given by algebraic compactness as the initial $F$-algebra \cite{barr} and corresponds to the denotation of the recursive type of natural numbers. Then, every value $c(l_1,\ldots,l_n)$ is naturally denoted by induction as a tree which has a root labelled by the symbol $c$, with a branch to every subtree $l_i$ for $1 \leq i \leq n$. We detail this construction in Section~\ref{sec:model}.

Finally, in order to denote the duplication/equality operator and the case expressions of Rfun, we require the notions of \textit{decidable equality} and \textit{decidable pattern matching}. Recall that in set theory, a set is \textit{decidable} (or \textit{has decidable equality}) whenever any pair of elements is either equal or different. Using the interpretation of restriction idempotents as propositions, we express decidability through the presence of \emph{complementary} restriction idempotents.

We conclude our study in Section~\ref{sec:categorical-models} with a definition of categorical models of reversible computing as join inverse rig categories with decidable equality and decidable pattern matching.

\section{Join inverse category theory}
\label{sec:join-inv}

While we assume prior knowledge of basic category theory, including monoidal categories and string diagrams, we briefly introduce some of the less well-known material on restriction and inverse categories (see also, e.g., \cite{cockett1,cockett2,cockett3,giles-thesis,guo-thesis}).

\begin{definition}
  A \emph{restriction structure} on a category consists of an operator 
  mapping each morphism $f : A \to B$ to a morphism $\ridm{f} : A \to A$ (called \emph{restriction idempotent}) such 
  that the following properties are satisfied:
  \begin{multicols}{2}
  \begin{enumerate}[(i)]
    \item $f \circ \ridm{f} = f$,
    \item $\ridm{f} \circ \ridm{g} = \ridm{g} \circ \ridm{f}$ for all $g : A \to C$,
    \item $\ridm{g \circ \ridm{f}} = \ridm{g} \circ \ridm{f}$ for all $g : A \to C$,
    \item $\ridm{g} \circ f = f \circ \ridm{g \circ f}$ for all $g : B \to C$.
  \end{enumerate}
  \end{multicols}
\end{definition}
A category with a restriction structure is called a \emph{restriction category}. A \textit{total} map is a morphism $f: A\to B$ such that $\ridm{f}=\text{id}_A$. Using an analogy from topology, the collection of restriction idempotents on an object $A$ is denoted $\mathcal{O}(A)$, and is sometimes even called the \emph{opens} on $A$. (Indeed, in the category $\cat{PTop}$ of topological spaces and partial continuous functions, the restriction idempotents coincide with the open sets.)
We recall now some basic properties of restriction idempotents:
\begin{lemma}
  In any restriction category, we have for all suitable $f$ and $g$ that
  \begin{multicols}{2}
  \begin{enumerate}[(i)]
    \item $\ridm{f} \circ \ridm{f} = \ridm{f}$,
    \item $\ridm{g \circ f} = \ridm{\ridm{g} \circ f}$,
    \item $\ridm{\ridm{g} \circ \ridm{f}} = \ridm{g} \circ \ridm{f}$, and
    \item $\ridm{g \circ f} \circ \ridm{f} = \ridm{g \circ f}$.
  \end{enumerate}
  \end{multicols}
\end{lemma}


As a trivial example, any category is a restriction category when equipped with the trivial restriction structure mapping $\ridm{f} = 1_A$ for all $f : A \to B$.

\begin{definition}
A morphism $f : A \to B$ in a restriction category is a partial isomorphism
whenever there exists a morphism $f^\dagger : B \to A$, the partial inverse of
$f$, such that $f^\dagger \circ f = \ridm{f}$ and $f \circ f^\dagger =
\ridm{f^\dagger}$.
\end{definition}

Note that the definite article -- \emph{the} partial inverse -- is justified, as partial inverses are unique whenever they exist.

\begin{definition}
An \emph{inverse category} is a restriction category in which every morphism is
a partial isomorphism.
\end{definition}

It is worth noting that this is not the only definition of an inverse category: historically, this mathematical structure has been defined as the categorical extension of inverse semigroups rather than as a particular class of restriction categories (see e.g.~\cite{kastl,lawson,hines-thesis}).

\begin{definition}
A \textit{zero object} in a restriction (or inverse) category is said to be a \emph{restriction zero} iff $\ridm{0_{A,A}} = 0_{A,A}$ for every zero endomorphism $0_{A,A}$.
\end{definition}

\begin{definition}
Parallel morphisms $f, g : A \to B$ of an inverse category are said to be \emph{inverse compatible}, denoted $f \asymp g$, if the following hold:
\begin{multicols}{2}
\begin{enumerate}[(i)]
  \item $g \circ \ridm{f} = f \circ \ridm{g}$ \hide{\emph{(restriction
  compatibility)}}
  \item $g^\dagger \circ \ridm{f^\dagger} = f^\dagger \circ \ridm{g^\dagger}$
   \hide{\emph{(corestriction compatibility)}}
\end{enumerate}
\end{multicols}
\end{definition}

By extension, one says that $S \subseteq \Hom(A,B)$ is inverse compatible if $s \asymp t$ for each $s,t \in S$. Note that \emph{disjoint} morphisms (those satisfying $g \circ \ridm{f} = 0_{A,B}$) are always inverse compatible. \\

The present work focuses on join inverse categories. The definition of such categories relies on the fact that in a restriction category $\cat{C}$, every hom-set $\cat{C}(A,B)$ gives rise to a poset when equipped with the following partial order: $f \leq g$ if and only if $g \circ \overline{f} = f$.

\begin{definition}[\cite{guo-thesis}]
\label{def:join_inv_cat}
An inverse category is a (countable) \emph{join inverse category} if it has a
restriction zero object, and satisfies that for all (countable) inverse
compatible subsets $S$ of all hom sets $\Hom(A,B)$, there exists a morphism
$\bigjoin_{s \in S} s$ such that
\begin{enumerate}[(i)]
\item $s \le \bigjoin_{s \in S} s$ for all $s \in S$, and $s \le t$ for all $s \in S$ implies $\bigjoin_{s \in S} s \le t$; \label{def:join_rest_sup}
\item $\ridm{\bigjoin_{s \in S} s} = \bigjoin_{s \in S} \ridm{s}$; \label{def:join_rest_cont}
\item $f \circ \left(\bigjoin_{s \in S} s \right) = \bigjoin_{s \in S}(f \circ s)$ for all $f : B \to X$; and
\item $\left(\bigjoin_{s \in S} s \right) \circ g = \bigjoin_{s \in S}(s \circ g)$ for all $g : Y \to A$.
\end{enumerate}
\end{definition}

On that matter, it is important to mention that there are significant mathematical results about join \emph{inverse} categories. In particular, there is an adjunction between the categories of join restriction categories and join inverse categories \cite{axelsen-kaarsgaard-fossacs16}.

In a join inverse category, every restriction idempotent $\ridm{e}$ has a \emph{pseudocomplement} $\ridm{e}^\perp$ given by $\bigvee_{\ridm{e'} \in C(\ridm{e})} \ridm{e'}$ where $C(\ridm{e}) = \{ \ridm{e'} \in \mathcal{O}(X) \mid \ridm{e'} \circ \ridm{e} = 0_{X,X}\}$. Indeed, the opens $\mathcal{O}(X)$ on any object $X$ form a frame (specifically a Heyting algebra) in any join inverse category (indeed, in any join restriction category, see \cite{cockett-garner}). Recall that in intuitionistic logic (which  Heyting algebras model), a proposition $p$ is \emph{decidable} when $p \lor \neg p$ holds. For our purposes, restriction idempotents with this property turn out to be very useful, and so we define them analogously:

\begin{definition}
A restriction idempotent $\ridm{e}$ is said to be \emph{decidable} if $\ridm{e} \vee \ridm{e}^\perp = \id$.
\end{definition}

Join inverse categories are canonically enriched in domains~\cite{axelsen-kaarsgaard-fossacs16,kaarsgaard-axelsen-glueck-jlamp}. This has the pleasant property that any \emph{morphism scheme}, i.e., Scott-continuous function $\Hom(X,Y) \to \Hom(X,Y)$ has a fixed-point:
\begin{theorem}[\cite{axelsen-kaarsgaard-fossacs16,kaarsgaard-axelsen-glueck-jlamp}]
In a join inverse category, any morphism scheme $\Hom(X,Y) \tot{\phi} \Hom(X,Y)$ has a least fixed point $X \tot{\fix \phi} Y$.
\end{theorem}
As one might imagine, this turns out to be useful in giving semantics to (systems of mutually) recursive functions. We note in this regard that with the canonical domain enrichment, both partial inversion $f \mapsto f^\dagger$ and the action of any join restriction functor on morphisms is continuous.

\newcommand{\Pfn}{\cat{Pfn}}

We conclude this section with a few examples:
The category $\cat{PInj}$ of sets and partial injections is a canonical example
of an inverse category (even further, by the categorical Wagner-Preston
theorem~\cite{cockett1}, every (locally small) inverse
category can be faithfully embedded in $\cat{PInj}$). For a partial injection
$f : A \to B$, define its restriction idempotent $\ridm{f} : A \to A$ by
$\ridm{f}(x) = x$ if $f$ is defined at $x$, and undefined otherwise. With this definition, every partial injection is a partial isomorphism. Moreover, the partial order on homsets corresponds to the usual partial order on partial functions: that is, for $f, g \in \cat{PInj}(A,B)$, $f \leq g$ if and only if, for every $x \in A$, $f$ is defined at $x$ implies that $g$ is defined at $x$ in such a way that $f(x)=g(x)$. Observe that \emph{every} restriction idempotent in $\cat{PInj}$ is decidable.

Another example of a join inverse categories is the category $\cat{PHom}$ of topological spaces and partial 
homeomorphisms with open range and domain of definition. Restrictions and joins are given as in $\cat{PInj}$, though showing that the join of partial homeomorphisms is again a partial homeomorphism requires use of the so-called \emph{pasting lemma}. In this category, restriction idempotents on a topological space correspond precisely to its open sets. As a consequence, the decidable restriction idempotents in $\cat{PHom}$ correspond to the \emph{clopen} sets, i.e., those simultaneously open and closed.

\section{Join inverse rig categories}
\label{sec:bimonoidal}

In this section, we provide the categorical foundations of join inverse rig categories, which are the basis of our model for reversible computing.

\begin{definition}[\cite{giles-thesis,axelsen-kaarsgaard-fossacs16}]
An (join) inverse category $\cat{C}$ with a restriction zero object $0$ is said to have a (join-preserving) \emph{disjointness tensor} if it is equipped with a symmetric monoidal (join-preserving) functor $- \oplus -$ (with
left unitor $\lambda_\oplus$, right unitor $\rho_\oplus$, associator
$\alpha_\oplus$, and commutator $\gamma_\oplus$) such that the restriction zero $0$ is tensor unit, and the canonical injections given by 
\[
\amalg_1 = (1_A \oplus 0_{0,B}) \circ \rho_\oplus^{-1} : A \to A \oplus B
\qquad
\text{ and }
\qquad
\amalg_2 = (0_{0,A} \oplus 1_B) \circ \lambda_\oplus^{-1} : B \to A \oplus B
\]
are jointly epic.
\end{definition}

At this point, to define the notion of an inverse product (which first appeared in \cite{giles-thesis}), we recall the definition of a $\dagger$-Frobenius semialgebra (see, \eg, \cite{giles-thesis}), that we later use to describe well-behaved products on inverse categories.

\begin{definition}
  In a (symmetric) monoidal $\dagger$-category, a \emph{$\dagger$-Frobenius semialgebra} is a pair $(X, \Delta_X)$ of an object $X$ and a map $\Delta_X : X \to X \otimes X$ such that the diagrams below commute.
\begin{center}
\begin{tikzpicture}
\node (XXX) {$X \otimes (X \otimes X)$};
\node (XXX') [right=15mm of XXX] {$(X \otimes X) \otimes X$};
\node (XX) [below of=XXX]{$X \otimes X$};
\node (XX') [below of=XXX'] {$X \otimes X$};
\node (phantom) [left=10mm of XX'] {\phantom{X}};
\node (X) [below=5mm of phantom] {$X$};

\draw[<-] (XXX) to node [morphism] {$\alpha$} (XXX');
\draw[<-] (XXX) to node [morphism,swap] {$\id_X \otimes \Delta_X$} (XX);
\draw[<-] (XXX') to node [morphism] {$\Delta_X \otimes \id_X$} (XX');
\draw[<-] (XX) to node [morphism,swap] {$\Delta_X$} (X);
\draw[<-] (XX') to node [morphism] {$\Delta_X$} (X);
\end{tikzpicture}
\begin{tikzpicture}
\node (X) {$X$};
\node (XX) [above=10mm of X, left of=X] {$X \otimes X$};
\node (XX') [below=10mm of X, right of=X] {$X \otimes X$};
\node (XXX) [above=10mm of X, right of=X] {$X \otimes (X \otimes X)$};
\node (XXX') [below=10mm of X, left of=X] {$(X \otimes X) \otimes X$};

\draw[<-] (XX) to node [morphism,swap] {$\Delta_X$} (X);
\draw[<-] (X) to node [morphism] {$\Delta_X^\dagger$} (XX');
\draw[<-] (XX) to node [morphism] {$\alpha \circ (\Delta_X^\dagger \otimes \id_X)$} (XXX);
\draw[<-] (XX) to node [morphism,swap] {$\alpha^{-1} \circ (\id_X \otimes \Delta_X^\dagger)$} (XXX');
\draw[<-] (XXX) to node [morphism] {$\id_X \otimes \Delta_X$ \ \phantom{$() \circ \alpha^{-1}$}} (XX');
\draw[<-] (XXX') to node [morphism,swap] {$\Delta_X \otimes \id_X$} (XX');
\end{tikzpicture}
\end{center}
Formally, the leftmost diagram (and its dual) makes $(X, \Delta_X)$ a cosemigroup,
and $(X, \Delta_X^\dagger)$ a semigroup, while the diagram to the right is called
the \emph{Frobenius condition}. One says that a $\dagger$-Frobenius semialgebra
$(X, \Delta_X)$ is \emph{special} if $\Delta_X^\dagger \circ \Delta_X = \id_X$, and
\emph{commutative} if the monoidal category in which it lives is symmetric and
$\gamma_{X,X} \circ \Delta_X = \Delta_X$ (where $\gamma_{X,X}$ is the symmetry of the monoidal category).
\end{definition}

Note that this definition is slightly nonstandard, in that $\dagger$-Frobenius algebras are more often defined in terms of a composition $X \to X \otimes X$ rather than a cocomposition $X \otimes X \to X$. This choice is entirely cosmetic, however, and has no bearing on the algebraic structure defined. Next, we recall the definition of an inverse product.

\begin{definition}
An inverse category $\cat{C}$ is said to have an \emph{inverse
product}~\cite{giles-thesis} if it is equipped with a symmetric monoidal
functor $- \otimes -$ (with
left unitor $\lambda_\otimes$, right unitor $\rho_\otimes$, associator
$\alpha_\otimes$, and commutator $\gamma_\otimes$) equipped with a natural
transformation $X \tot{\Delta} X \otimes X$ such that the pair $(X, \Delta_X)$ is a
special and commutative $\dagger$-Frobenius semialgebra for any object $X$.
\end{definition}

In the context of inverse products, we think of the cosemigroup composition $\Delta : X \to X \otimes X$ as a \emph{duplication} map, and its inverse as a (partial) \emph{equality test} map defined precisely on pairs of equal things. In this way, categories with inverse products are really inverse categories with robust notions of duplication and partial equality testing.

When clear from the context, we omit the subscripts on unitors,
associators, and commutators. We are finally ready to define join inverse rig categories.

\begin{definition}
\label{def:bimonoidal}
A join inverse rig category is a join inverse category equipped
with a join-preserving inverse product $(\otimes,1)$ and a join-preserving disjointness tensor 
$(\oplus,0)$, such that there are natural isomorphisms 
$X \otimes (Y \oplus Z) \xrightarrow{\delta_L} (X \otimes Y) \oplus (X \otimes Z)$
and
$(X \oplus Y) \otimes Z \xrightarrow{\delta_R} (X \otimes Z) \oplus (Y \otimes Z)$
(the \emph{distributors}) natural in $X$, $Y$, and $Z$, and natural isomorphisms 
$0 \otimes X \xrightarrow{\nu_L} 0$ and $X \otimes 0 \xrightarrow{\nu_R} 0$ (the \emph{annihilators}) natural in $X$, such that $(\otimes,\oplus,0,1)$ form a rig category in the usual sense.
\end{definition}

Recall that a category $\cat{C}$ is algebraically compact for a class $\mathcal{L}$ of endofunctors on $\cat{C}$ if every endofunctor $F$ in the class $\mathcal{L}$ has a unique fixpoint $\mu X. FX$ (or shortly $\mu F$) given by the initial $F$-algebra.

\begin{theorem}\label{thm:bimonoidal_embedding}
  Any locally small inverse category can be faithfully embedded in a category of which is
  \begin{enumerate}[(i)]
    \item a join inverse rig category,
    \item algebraically compact for join inverse-preserving functors.
  \end{enumerate}
\end{theorem}

\begin{proof}
  $\cat{PInj}$ has all of these properties (see also \cite{axelsen-kaarsgaard-fossacs16}), and any locally small inverse category embeds faithfully into it by the categorical Wagner-Preston theorem \cite[Prop.~3.11]{heunen}.
\end{proof}

In light of this, one can assume without loss of generality that our join inverse rig categories from here on out are algebraically compact for restriction endofunctors, so that every restriction endofunctor $F$ has a unique fixpoint. Note that this provides the first steps towards an interpretation of recursive types in Idealized Theseus~\cite{theseus}.



On a final note, our definition of join inverse rig categories is fairly similar to the notion of distributive join inverse category in Giles' terminology \cite[Sec.~9.2]{giles-thesis}.

\section{A primer in Rfun}
The syntax of the programming language Rfun can be summarised by the following grammars:
\begin{align*}
\text{Left Expressions } l &::= x \mid c(l_1, \ldots, l_n) \mid \lfloor l \rfloor \\
\text{Expressions } e ::= l  &\mid \textbf{let } l_{out} = f\, l_{in} \textbf{ in } e\\ &\mid \textbf{rlet } l_{in} = f\, l_{out} \textbf{ in } e\\ &\mid \textbf{case}\, l \textbf{ of } \{l_i \to e_i\}_{i=1}^m\\
\text{Definitions } d &::= f\, x \triangleq e\\
\text{Programs } q &::= d_1;\ldots;d_n
\end{align*}

\newcommand{\xjudg}[4]{\langle #1,#2\rangle \vdash #3 \Downarrow #4}
\newcommand{\judg}{\xjudg{q}{\sigma}{e}{v}}
\newcommand{\judgnov}{\xjudg{q}{\sigma}{e}{}}
\newcommand{\leftexparr}{\vdash}

The \emph{values} of Rfun are of the form $c(v_1, \dots, v_n)$ for $n \ge 0$, where $c$ is some constructor name, and each $v_i$ is a value. A value of the form $c()$ is called a \emph{symbol}, and is typically written simply as $c$. Tuples $\langle l_1, \dots, l_n \rangle$ are also permitted, though these are taken to be syntactic sugar for $\langle\rangle (l_1, \dots, l_n)$ where $\langle\rangle$ is a distinguished constructor name. Note that,
contrary to the original presentation of the language, we make the simplifying assumption that the argument to a function definition is always a variable, and not a (more general) left expression. We recover the original expressivity of Rfun by introducing some syntactic sugar: definitions $f\, l \triangleq e$ stand for terms $$f\, x \triangleq \textbf{case } x \textbf{ of } l \to e.$$

It is important to recall before going any deeper in the presentation of Rfun that:
\begin{itemize}
 \item We assume three distinct, denumerable sorts for variables, constructor names, and functions.
 \item We suppose that programs in the same sequences of definitions have (pairwise) distinct functional identifiers. 
 \item Variables may appear only once in left expressions, and may be used only once in expressions (linearity).
 \item Domains of substitutions are (pairwise) disjoint.
\end{itemize}

Now a presentation of Rfun's big step operational semantics is  given, with expression judgement $\judg$ instead of the notation $\sigma \vdash_q e \hookrightarrow v$ from \cite{yokoyama-axelsen-glueck-rc2011}. Concretely, the pair of a program $q$ and a substitution (i.e. partial function) $\sigma$ constitutes a programming context $\langle q, \sigma\rangle$. Then, the expression judgement $\judg$ means that the expression $e$ evaluates to the value $v$ in the context $\langle q, \sigma\rangle$. Let us write $\judgnov$ when there is some value $v$ such that $\judg$.

As for the pattern matching operations which guide the formation of subtitutions, we replace $v \triangleleft l \rightsquigarrow \sigma$ \cite[Fig.~3, pp.~19]{yokoyama-axelsen-glueck-rc2011} by the more restrictive statement $\xjudg{q}{\sigma}{l}{v}$. The relation between those two expressions is given by the following correspondence:
$$
\begin{prooftree}
v \triangleleft l \rightsquigarrow \sigma
\Justifies
\forall q.\, \judg
\end{prooftree}
$$
\ \\
This leads us to the following operational semantics, which guarantees that computations are reversible (see \cite{yokoyama-axelsen-glueck-rc2011}). Note in particular the distinction between $\mathbf{let}$ and $\mathbf{rlet}$-expressions, which are used to call functions in the \emph{forward} and \emph{backward} directions respectively.

$$
\begin{prooftree}
\justifies
\langle q, \{x \mapsto v\}\rangle \vdash x \Downarrow v
\end{prooftree}
\qquad
\begin{prooftree}
\lfloor v \rfloor \downarrow = v'\qquad \xjudg{q}{\sigma}{l}{v'}
\justifies
\xjudg{q}{\sigma}{\lfloor l \rfloor}{v}
\end{prooftree}
$$
\ \\
$$
\begin{prooftree}
\langle q, \sigma_1\rangle \vdash l_1 \Downarrow v_1 \quad\cdots\quad \langle q, \sigma_n\rangle \vdash l_n \Downarrow v_n
\justifies
\langle q, \uplus_{i=1}^n \sigma_i\rangle \vdash c(l_1,\ldots,l_n) \Downarrow c(v_1,\ldots,v_n)
\end{prooftree}
\qquad
\begin{prooftree}
f\, x_f \triangleq e_f \in q \quad \xjudg{q}{\sigma}{x}{v'}
\justifies
\begin{prooftree} 
\xjudg{q}{\sigma_f}{x_f}{v'} \quad \xjudg{q}{\sigma_f}{e_f}{v}
\justifies
\xjudg{q}{\sigma}{f\, x}{v}
\end{prooftree}
\thickness=0pt
\end{prooftree}
$$
\ \\
$$
\begin{prooftree}
\xjudg{q}{\sigma_\text{in}}{f\, l_\text{in}}{v_\text{out}} \qquad \xjudg{q}{\sigma_\text{out} \uplus \sigma_e}{e}{v}
\justifies
\begin{prooftree}
\xjudg{q}{\sigma_\text{out}}{l_\text{out}}{v_\text{out}}
\justifies
\xjudg{q}{\sigma_\text{in} \uplus \sigma_e}{\textbf{let}\, l_\text{out} = f\, l_\text{in} \textbf{ in } e}{v}
\end{prooftree}
\thickness=0pt
\end{prooftree}
$$
\ \\
$$
\begin{prooftree}
\xjudg{q}{\sigma_\text{out}}{f\, l_\text{out}}{v_\text{in}} \qquad \xjudg{q}{\sigma_\text{out} \uplus \sigma_e}{e}{v}
\justifies
\begin{prooftree}
\xjudg{q}{\sigma_\text{in}}{l_\text{in}}{v_\text{in}}
\justifies
\xjudg{q}{\sigma_\text{in} \uplus \sigma_e}{\textbf{rlet}\, l_\text{in} = f\, l_\text{out} \textbf{ in } e}{v}
\end{prooftree}
\thickness=0pt
\end{prooftree}
\qquad
\begin{prooftree}
\xjudg{q}{\sigma_l}{l}{v'} \qquad \xjudg{q}{\sigma_{l_j} \uplus \sigma_t}{e_j}{v}
\justifies
\begin{prooftree}
j=\text{min}\{i \mid \forall q.\, \xjudg{q}{\sigma_{l_i}}{l_i}{v'}\}
\justifies
\begin{prooftree}
= \text{min}\{i \mid \forall q.\, l' \in \text{leaves}(e_i) \wedge \xjudg{q}{-}{l'}{v}\}
\justifies
\langle q, \sigma_l \uplus \sigma_t \rangle \vdash_q \textbf{case } l \textbf{ of } \{l_i \to e_i\}_{i=1}^m \Downarrow v
\end{prooftree}
\thickness=0pt
\end{prooftree}
\thickness=0pt
\end{prooftree}
$$

\section{A categorical model of Rfun}
\label{sec:model}
To give a model of Rfun, we start with a join inverse rig category  and provide
\begin{enumerate}
    \item a construction of \emph{values} as an object $T(S)$ over a given alphabet $S$ (thought of as an alphabet of symbols),
    \item an interpretation of open left expressions with $k$ free variables as morphisms $T(S)^{\otimes k} \to T(S)$,
    \item an interpretation of open expressions with $k$ free variables as morphisms $T(S)^{\otimes k} \to T(S)$ in a \emph{program context},
    \item an interpretation of function definitions as open terms
    with a single free variable, and
    \item an interpretation of programs as a sum of function definitions wrapped in a fixed point (the \emph{program context}).
\end{enumerate}

\subsection{Values}
We start by constructing a denumerable object $S$ of \emph{symbols}, each identified by a unique morphism $1 \to S$ (where $1$ is unit of the inverse product). Since the sort of symbols is denumerable by assumption, provided that we can construct such an object, we can uniquely identity a symbol $s$ with a morphism $1 \to S$. Straightforwardly, we define $S$ to be the least fixed point of the join restriction functor $N(X) = 1 \oplus X$, via algebraic compactness: This yields an isomorphism $S \tot{\unfold_S} S \oplus 1$ (by Lambek's lemma) with inverse $S \oplus 1 \tot{\fold_S} S$. This allows us to identity the first symbol $s_1$ with $1 \tot{\amalg_1} 1 \oplus S \tot{\fold_S} S$, the second symbol $s_2$ with $1 \tot{\amalg_1} 1 \oplus (1 \oplus S) \tot{\id \oplus \fold_S} 1 \oplus S \tot{\fold_S} S$, and so on. For this reason, we will simply write $1 \tot{s} S$ for the morphism corresponding to the symbol $s$. Note that this has a partial inverse $S \tot{s^\dagger} 1$, which we think of as a corresponding \emph{assertion} that a given symbol is precisely $s$.

With this in hand, we can construct the object $T(S)$ of Rfun values, where the functor $T$ is defined as follows:
$$T(X) = \mu K. X \otimes L(K) \qquad \qquad L(X) = \mu K. 1 \oplus (X \otimes K)$$

Intuitively, $L$ maps an object $X$ to that of lists of $X$, while $T$ maps an object $X$ to nonempty finite trees with $X$-values at each node. In $\cat{PInj}$, a few examples of elements of $T(X)$ (for some set $X$ and $a,b,c \in X$) are shown in the figure below.

\begin{wrapfigure}{r}{0.4\textwidth}
\begin{center}
\begin{tikzpicture}
\node (c) {\footnotesize $c$};

\node (b') [right=15mm of c] {$b$};
\node (c') [below=8mm of b'] {$c$};

\node (a2) [right=20mm of b'] {$a$};
\node (a3) [below=8mm of a2] {$a$};
\node (b1) [left=8mm of a3] {$b$};
\node (b2) [right=8mm of a3] {$b$};
\node (c2) [below=8mm of a3] {$c$};
\node (a4) [right=3mm of c2] {$a$};
\node (b3) [right=6mm of a4] {$b$};

\draw[-] (a2) to node {} (b1);
\draw[-] (a2) to node {} (b2);
\draw[-] (a2) to node {} (a3);
\draw[-] (a3) to node {} (c2);
\draw[-] (b2) to node {} (a4);
\draw[-] (b2) to node {} (b3);

\draw[-] (b') to node {} (c');

\node (l1) [below=28mm of c] {\scriptsize $(1)$};
\node (l2) [below=28mm of b'] {\scriptsize $(2)$};
\node (l3) [below=28mm of a2] {\scriptsize $(3)$};
\end{tikzpicture}
\end{center}
\end{wrapfigure}

The object $T(S)$ allows us to represent terms in Rfun very naturally, namely by their syntax trees. As such, $(1)$ above corresponds to the value $c$, $(2)$ to $b(c)$, and $(3)$ to the value $a(b,a(c),b(a,b))$. Though this is often how untyped programming languages are modelled, we do not formally require $T(S)$ to be an universal object of the category, as long as it is rich enough for $\Hom(1, T(S))$ to uniquely encode all values.

As in the case for $S$, we have isomorphisms $L(X) \tot{\unfold_L} 1 \oplus (X \otimes L(X))$ (with inverse $\mathrm{fold}_L$) and $T(X) \tot{\unfold_T} X \otimes L(T(X))$ (with inverse $\mathrm{fold}_T$) by Lambek's lemma.
The object $L(X)$ can be thought of as lists over $X$: The morphism $1 \tot{\amalg_1} 1 \oplus (X \otimes L(X)) \tot{\mathrm{fold}_L} L(X)$ is thought of as the empty list $[]$, and $X \otimes L(X) \tot{\amalg_2} 1 \oplus (X \otimes L(X)) \tot{\mathrm{fold}_L} L(X)$ as the usual \emph{cons} operation. By combining these, we obtain an inductive family of morphisms $X^{\otimes n} \tot{\mathrm{pack}_n} L(X)$ mapping $n$-ary tuples into lists, with $\mathrm{pack}_0$ given by $1 \tot{[]} L(X)$, and $\mathrm{pack}_{n+1}$ by $X^{\otimes n+1} \tot{\cong} X \otimes X^{\otimes n} \tot{\id \otimes \mathrm{pack}_n} X \otimes L(X) \tot{\mathrm{cons}} L(X)$. Their partial inverses, $L(X) \tot{\mathrm{unpack}_n} X^{\otimes n}$, can be thought of as unpacking lists of length precisely $n$ into an $n$-ary tuple, being undefined on lists of any other length. In particular, one can show that $\mathrm{unpack}_n$ and $\mathrm{unpack}_m$ are disjoint, so that $\ridm{\mathrm{unpack}_n} \circ \ridm{\mathrm{unpack}_m} = 0_{L(X),L(X)}$ whenever $n \neq m$.

\subsection{Left expressions and patterns}
In order to continue with the construction of left expressions, we first need to make one of two assumptions about decidability. Say that an object has \emph{decidable equality} if the restriction idempotent $\ridm{\Delta_X^\dagger}$ is decidable.

\begin{assumption}
$T(S)$ has decidable equality.
\end{assumption}

We justify this terminology by the fact that the cocomposition $X \tot{\Delta_X} X \otimes X$ is thought of as \emph{duplication}. As such, $\Delta_X^\dagger$ is only ever defined for \emph{results} of duplication, i.e., points which are equal. This assumption allows us to define the \emph{duplication/equality} operator on $T(S)$ as shown in Figure~\ref{fig:leftexp_sem}. The morphism is described using string diagrams read from bottom to top, with parallel wires representing inverse products.
\begin{figure}
    \centering
    \includegraphics[width=\textwidth]{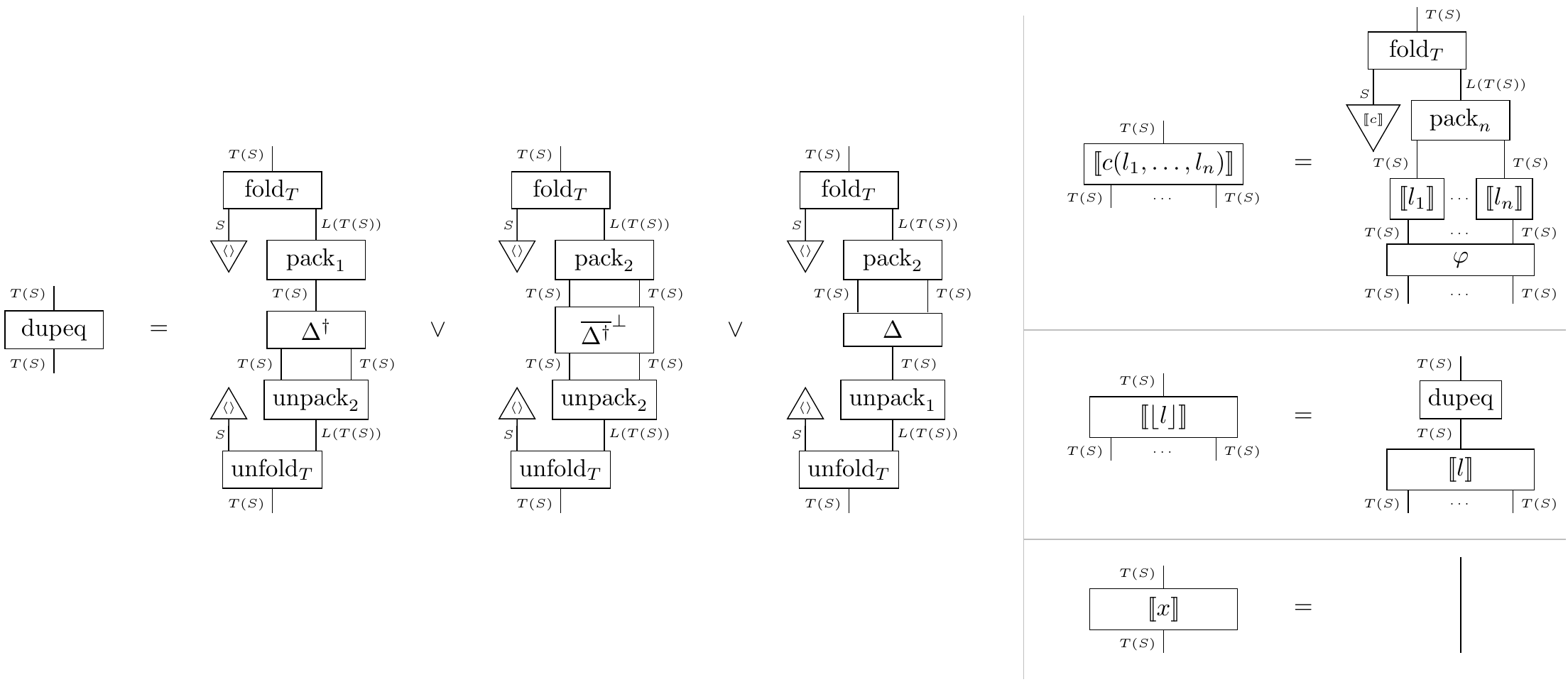}
    \caption{The semantics of left expressions.}
    \label{fig:leftexp_sem}
\end{figure}   
The three morphisms that join to form the definition of $T(S) \tot{\mathrm{dupeq}} T(S)$ correspond to the three cases in the definition of duplication/equality in Figure~\ref{fig:dupeq}: The first corresponds to the case where $\lfloor \langle x,y \rangle \rfloor = \langle x \rangle$ when $x=y$, the second to $\lfloor \langle x,y \rangle \rfloor = \langle x,y \rangle$ when $x \neq y$, and the third to $\lfloor \langle x \rangle \rfloor = \langle x,x \rangle$. That this join exists at all follows by the fact that these morphisms are pairwise disjoint (the first and second morphism are both disjoint from the third since $\mathrm{unpack}_1$ and $\mathrm{unpack}_2$ are disjoint; and the two first morphisms are disjoint since $\Delta^\dagger$ and $\ridm{\Delta^\dagger}^\perp$ are disjoint). Note the use of the symbol $\langle \rangle$, representing the fact that tuples $\langle l_1, \dots, l_n \rangle$ in Rfun are simply syntactic sugar for $\langle\rangle (l_1, \dots, l_n)$ with $\langle\rangle$ a distinguished symbol. Interestingly, from this definition, it is straightforward to show that $\mathrm{dupeq}^\dagger = \mathrm{dupeq}$.

To give a semantics to constructed terms and variables, we start with the idea that the free variables of a (left) expression are interpreted as wires of $T(S)$ type going into the denotation, and that the result of a denotation is again of $T(S)$ type. As such, a (left) expression with $k$ free variables is interpreted as a morphism $T(S)^{\otimes k} \to T(S)$. 

The semantics of open left expressions is shown in Figure~\ref{fig:leftexp_sem}. Note the permutation $\varphi$ of variables wires at the beginning of $[\![{c(l_1, \dots, l_n) ]\!]}$. This is necessary since free variables may not be used in the order they are given (consider, e.g., the program that takes any tuple $\langle x,y \rangle$ as input and returns $\langle y,x \rangle$ as output), so some reordering must take place first. Since a variable expression $x$ contains precisely one free variable, and nothing is required to further prepare it for use, its semantics is simply the identity. As expected, the semantics of a duplication/equality expression is simply given by handing off the semantics of the inner left expression to the duplication/equality operator previously constructed.

The semantics of left expressions shown in Figure~\ref{fig:leftexp_sem} concern their use as \emph{expressions}, but left expressions can also be used as \emph{patterns} in pattern matching case expressions. Fortunately, the interpretation of a left expression as a pattern is simply the partial inverse to its interpretation as an expression. Partiality is key to this insight: Since a left expression describes the formation of a value of a given form, its partial inverse will only ever be \emph{defined} on values of that form. Further, while an open expression consumes (parts of) variable bindings in order to produce a value, patterns do the opposite, consuming a value to produce a variable binding which binds subvalues to free variables. This is reflected in the fact that the types of a left expression as a pattern is then $T(S) \to T(S)^{\otimes k}$, i.e., a partial map which, if it succeeds, splits a value into subvalues accordingly.

\subsection{Expressions}
Before we're able to proceed with a semantics for expressions, we first need to make our second and final assumption regarding decidability, here involving pattern matching. We say that $T(S)$ has \emph{decidable pattern matching} if for any left expression $l$ the restriction idempotent $\ridm{\sem{l}^\dagger}$ is decidable.

\begin{assumption}
$T(S)$ has decidable pattern matching.
\end{assumption}

We now turn to the semantics for expressions which, unlike those for left expressions, need to be parametrised by a \emph{program context}. This is necessary since the expressions include (inverse) function calls, the semantics of which will naturally differ according to the definition of the function being invoked. Similar to \cite{glueck-kaarsgaard-yokoyama}, we take a program context of $n$ functions to be a morphism $T(S)^{\oplus n} \to T(S)^{\oplus n}$, and think of it as a sum $\sem{f_1} \oplus \cdots \oplus \sem{f_n}$ of interpretations of constituent functions.

Given such a program context $\xi$, we will use $\xi_i$ as shorthand for $\amalg_i^\dagger \circ \xi \circ \amalg_i$, and since $(\xi^\dagger)_i = (\xi_i)^\dagger$ we may also use $\xi_i^\dagger$ unambiguously. The semantics of expressions in a program context are shown in Figure~\ref{fig:exp_sem}. As can be seen, function calls to the $i^{th}$ function in the program context are handled by evaluating the input before handing it off to the $i^{th}$ component of the program context. The \textbf{let} part of these expressions is handled by first permuting (reflecting the fact that some of the free variables may be used in $l_\mathit{in}$ and others in $e$), and then passing the result to the semantics of $e$ as the contents of a fresh variable. Inverse function invocation (using an \textbf{rlet} binder) is handled analogously, though using the \emph{partial inverse} to the program context, rather than the program context itself. Left expressions are handled by passing them on to the earlier definition in Figure~\ref{fig:leftexp_sem}.

The semantics of \textbf{case} expressions require special attention: As with function invocation, since $l$ may only use some of the free variables, we must first select the ones used by $l$ using the permutation $\varphi$, and then pass the rest on to the bodies of each branch (which, by linearity, are each required to use the remaining variables). Then, after evaluating $l$, for each branch $l_i \to e_i$, we need to ensure that only values that did not match any of the previous branches are fed to this branch. This makes sure that branches are tried in the given order, and is why we must compose with $\ridm{\sem{l_{i-1}}^\dagger}^\perp \circ \cdots \circ \ridm{\sem{l_1}^\dagger}^\perp$ before trying to match using $\sem{l_i}^\dagger$. Should the match succeed, the resulting binding is passed to the semantics of the branch body $e_i$. In this way, each branch of the \textbf{case} expression is constructed as a partial map performing its own pattern matching and evaluation, and the meaning of the entire case expression is simply given by gluing these partial maps together using the join. That the join exists follows by the fact that each branch is made explicitly disjoint from all of the previous ones, by composition with $\ridm{\sem{l_{i-1}}^\dagger}^\perp \circ \cdots \circ \ridm{\sem{l_1}^\dagger}^\perp$ (not unlike Gram-Schmidt orthogonalisation).

\begin{figure}
    \centering
    \includegraphics[width=\textwidth]{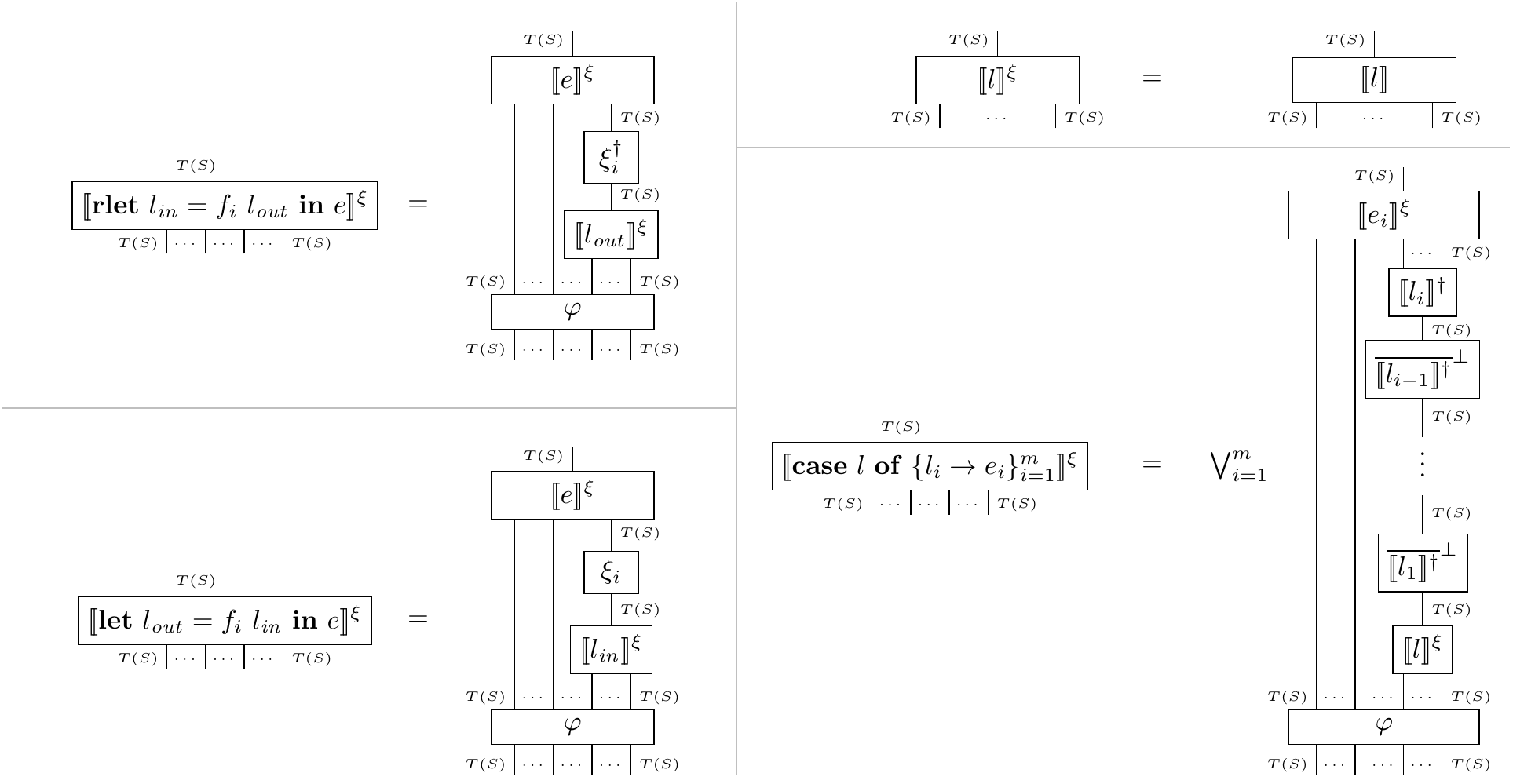}
    \caption{The semantics of expressions in a program context $\xi$.}
    \label{fig:exp_sem}
\end{figure}

\subsection{Programs}
We are finally ready to take on the semantics of function definitions and programs in Rfun. This is comparatively much simpler. Like expressions, the semantics of function definitions is given parametrised by a program context $\xi$. Since we made the simplifying assumption that any function definition is of the form $f\, x \triangleq e$ for a single variable $x$, $e$ alone is an expression with exactly one free variable. As such, we simply define the meaning of a function definition to be given by the semantics of its body,
$$
\sem{f\, x \triangleq e}^\xi = \sem{e}^\xi \enspace.
$$
Finally, a program is simply a list of function definitions, but unlike functions, their semantics should be self-contained and not depend on context. To solve this bootstrapping problem, we use the fixed point operator on continuous morphism schemes $\Hom(T(S)^{\oplus n}, T(S)^{\oplus n}) \to \Hom(T(S)^{\oplus n}, T(S)^{\oplus n})$ given by the canonical enrichment in domains. This, importantly, also allows (systems of mutually) recursive functions to be defined.
$$
\sem{d_1; \dots, d_n} = \fix (\xi \mapsto \sem{d_1}^\xi \oplus \cdots \oplus \sem{d_n}^\xi) \enspace.
$$
This definition requires us to verify that the function $\xi \mapsto \sem{d_1}^\xi \oplus \cdots \oplus \sem{d_n}^\xi$ is always continuous, but this follows straightforwardly from the observation that only continuous operations (in particular vertical and horizontal composition) are ever performed on the program context $\xi$.

\subsection{Other reversible languages}
\label{sub:other-revs}

Join inverse categorie, as a model for reversible computing, also outline models for other reversible languages than Rfun. We have already noted in Section~\ref{sec:bimonoidal} that Theorem~\ref{thm:bimonoidal_embedding} leads to an interpretation of recursive types in Theseus~\cite{theseus}. Noting that Theseus is built on the reversible combinator calculus $\Pi^0$~\cite{james-sabry-infeff}, and that it can be straightforwardly shown that join inverse rig categories are algebraically compact over join restriction functors and that they are examples of $\dagger$-traced $\omega$-continuous rig categories (in the sense of Karvonen \cite{karvonen}), it follows that join inverse rig categories are models of Theseus. The fact that join inverse rig categories are $\dagger$-traced was established in \cite{kaarsgaard-axelsen-glueck-jlamp}:

\begin{proposition}
Every join inverse category with a join-preserving disjointness tensor (specifically any join inverse rig category) has a uniform trace operator $$\mathrm{Tr}_{A,B}^U : \Hom(A \oplus U, B \oplus U) \to \Hom(A,B)$$
which satisfies $\mathrm{Tr}_{A,B}^U(f)^\dagger = \mathrm{Tr}_{B,A}^U(f^\dagger)$.
\end{proposition}

Join inverse rig categories also turn out to constitute a model for structured reversible flowcharts, as studied in~\cite{revflowcharts}. The two crucial elements in the interpretation of reversible flowcharts in inverse categories~\cite{revflowcharts} are: (inverse) \emph{extensivity}, which holds for any join inverse category with a join-preserving disjointness tensor (specifically any join inverse rig category) and gives semantics to \emph{reversible control flow}; and the presence of the $\dagger$-trace constructed previously, which describes \emph{reversible tail recursion}.

\section{Concluding remarks}
\label{sec:categorical-models}
In summary, we have introduced join inverse categories and constructed the categorical semantics of the expressions of the language Rfun. We have also argued that our categorical framework fits neatly in other reversible languages (Theseus, reversible flowcharts). With arguably weak categorical assumptions, we showcase the strengths of join inverse category theory in the study of the semantics of reversible programming. 


Rig categories and groupoids have previously been considered in connection with reversible computing. Notably, the $\Pi$ family~\cite{bowman-james-sabry-dagger,james-sabry-infeff} of reversible programming languages, as well as the language CoreFun~\cite{jacobsen-kaarsgaard-thomsen-corefun}, are essentially term languages for dagger rig categories. Many extensions of $\Pi$ have since been considered (see, e.g., \cite{carette-sabry-weakriggrpds,chen-sabry-negative-fractional,heunen-kaarsgaard-qie,kaarsgaard-veltri-engarde}). The notion of a \emph{distributive} inverse category~\cite{giles-thesis} is also strongly related to our approach.

As future work, we consider the categorical treatment of languages for reversible circuits, for example in the context of abstract embedded circuit-description languages such as EWire~\cite{ewire}. Such a study would be particularly relevant in the context of the development of verification and optimisation tools for Field-Programmable Gate Array (FPGA) circuits, but also quantum circuits.

\bibliographystyle{eptcs}
\bibliography{revcat}

\end{document}

%% file: notations.tex
\newcommand{\fix}{\operatorname{fix}}

\newcommand{\cat}[1]{\mathbf{#1}}

\newcommand{\Dcpo}{\mathbf{Dcpo}}

\newcommand{\opp}[1]{#1^\mathbf{op}}

\newcommand{\id}{\operatorname{id}}

\newcommand{\hide}[1]{}

%% file: prooftree.tex
\newdimen\proofrulebreadth \proofrulebreadth=.05em
\newdimen\proofdotseparation \proofdotseparation=1.25ex
\newdimen\proofrulebaseline \proofrulebaseline=2ex
\newcount\proofdotnumber \proofdotnumber=3
\let\then\relax
\def\hfi{\hskip0pt plus.0001fil}
\mathchardef\squigto="3A3B
\newif\ifinsideprooftree\insideprooftreefalse
\newif\ifonleftofproofrule\onleftofproofrulefalse
\newif\ifproofdots\proofdotsfalse
\newif\ifdoubleproof\doubleprooffalse
\let\wereinproofbit\relax
\newdimen\shortenproofleft
\newdimen\shortenproofright
\newdimen\proofbelowshift
\newbox\proofabove
\newbox\proofbelow
\newbox\proofrulename
\def\shiftproofbelow{\let\next\relax\afterassignment\setshiftproofbelow\dimen0 }
\def\shiftproofbelowneg{\def\next{\multiply\dimen0 by-1 }%
\afterassignment\setshiftproofbelow\dimen0 }
\def\setshiftproofbelow{\next\proofbelowshift=\dimen0 }
\def\setproofrulebreadth{\proofrulebreadth}

\def\prooftree{
\ifnum  \lastpenalty=1
\then   \unpenalty
\else   \onleftofproofrulefalse
\fi
\ifonleftofproofrule
\else   \ifinsideprooftree
        \then   \hskip.5em plus1fil
        \fi
\fi
\bgroup
\setbox\proofbelow=\hbox{}\setbox\proofrulename=\hbox{}%
\let\justifies\proofover\let\leadsto\proofoverdots\let\Justifies\proofoverdbl
\let\using\proofusing\let\[\prooftree
\ifinsideprooftree\let\]\endprooftree\fi
\proofdotsfalse\doubleprooffalse
\let\thickness\setproofrulebreadth
\let\shiftright\shiftproofbelow \let\shift\shiftproofbelow
\let\shiftleft\shiftproofbelowneg
\let\ifwasinsideprooftree\ifinsideprooftree
\insideprooftreetrue
\setbox\proofabove=\hbox\bgroup$\displaystyle 
\let\wereinproofbit\prooftree
\shortenproofleft=0pt \shortenproofright=0pt \proofbelowshift=0pt
\onleftofproofruletrue\penalty1
}

\def\eproofbit{
\ifx    \wereinproofbit\prooftree
\then   \ifcase \lastpenalty
        \then   \shortenproofright=0pt  
        \or     \unpenalty\hfil         
        \or     \unpenalty\unskip       
        \else   \shortenproofright=0pt  
        \fi
\fi
\global\dimen0=\shortenproofleft
\global\dimen1=\shortenproofright
\global\dimen2=\proofrulebreadth
\global\dimen3=\proofbelowshift
\global\dimen4=\proofdotseparation
\global\count255=\proofdotnumber
$\egroup  
\shortenproofleft=\dimen0
\shortenproofright=\dimen1
\proofrulebreadth=\dimen2
\proofbelowshift=\dimen3
\proofdotseparation=\dimen4
\proofdotnumber=\count255
}

\def\proofover{
\eproofbit 
\setbox\proofbelow=\hbox\bgroup 
\let\wereinproofbit\proofover
$\displaystyle
}%
\def\proofoverdbl{
\eproofbit 
\doubleprooftrue
\setbox\proofbelow=\hbox\bgroup 
\let\wereinproofbit\proofoverdbl
$\displaystyle
}%
\def\proofoverdots{
\eproofbit 
\proofdotstrue
\setbox\proofbelow=\hbox\bgroup 
\let\wereinproofbit\proofoverdots
$\displaystyle
}%
\def\proofusing{
\eproofbit 
\setbox\proofrulename=\hbox\bgroup 
\let\wereinproofbit\proofusing
\kern0.3em$
}

\def\endprooftree{
\eproofbit 
  \dimen5 =0pt
\dimen0=\wd\proofabove \advance\dimen0-\shortenproofleft
\advance\dimen0-\shortenproofright
\dimen1=.5\dimen0 \advance\dimen1-.5\wd\proofbelow
\dimen4=\dimen1
\advance\dimen1\proofbelowshift \advance\dimen4-\proofbelowshift
\ifdim  \dimen1<0pt
\then   \advance\shortenproofleft\dimen1
        \advance\dimen0-\dimen1
        \dimen1=0pt
        \ifdim  \shortenproofleft<0pt
        \then   \setbox\proofabove=\hbox{%
                        \kern-\shortenproofleft\unhbox\proofabove}%
                \shortenproofleft=0pt
        \fi
\fi
\ifdim  \dimen4<0pt
\then   \advance\shortenproofright\dimen4
        \advance\dimen0-\dimen4
        \dimen4=0pt
\fi
\ifdim  \shortenproofright<\wd\proofrulename
\then   \shortenproofright=\wd\proofrulename
\fi
\dimen2=\shortenproofleft \advance\dimen2 by\dimen1
\dimen3=\shortenproofright\advance\dimen3 by\dimen4
\ifproofdots
\then
        \dimen6=\shortenproofleft \advance\dimen6 .5\dimen0
        \setbox1=\vbox to\proofdotseparation{\vss\hbox{$\cdot$}\vss}%
        \setbox0=\hbox{%
                \advance\dimen6-.5\wd1
                \kern\dimen6
                $\vcenter to\proofdotnumber\proofdotseparation
                        {\leaders\box1\vfill}$%
                \unhbox\proofrulename}%
\else   \dimen6=\fontdimen22\the\textfont2 
        \dimen7=\dimen6
        \advance\dimen6by.5\proofrulebreadth
        \advance\dimen7by-.5\proofrulebreadth
        \setbox0=\hbox{%
                \kern\shortenproofleft
                \ifdoubleproof
                \then   \hbox to\dimen0{%
                        $\mathsurround0pt\mathord=\mkern-6mu%
                        \cleaders\hbox{$\mkern-2mu=\mkern-2mu$}\hfill
                        \mkern-6mu\mathord=$}%
                \else   \vrule height\dimen6 depth-\dimen7 width\dimen0
                \fi
                \unhbox\proofrulename}%
        \ht0=\dimen6 \dp0=-\dimen7
\fi
\let\doll\relax
\ifwasinsideprooftree
\then   \let\VBOX\vbox
\else   \ifmmode\else$\let\doll=$\fi
        \let\VBOX\vcenter
\fi
\VBOX   {\baselineskip\proofrulebaseline \lineskip.2ex
        \expandafter\lineskiplimit\ifproofdots0ex\else-0.6ex\fi
        \hbox   spread\dimen5   {\hfi\unhbox\proofabove\hfi}%
        \hbox{\box0}%
        \hbox   {\kern\dimen2 \box\proofbelow}}\doll%
\global\dimen2=\dimen2
\global\dimen3=\dimen3
\egroup 
\ifonleftofproofrule
\then   \shortenproofleft=\dimen2
\fi
\shortenproofright=\dimen3
\onleftofproofrulefalse
\ifinsideprooftree
\then   \hskip.5em plus 1fil \penalty2
\fi
}

%% file: rfun_fib.tex
\footnotesize
\begin{align*}
  \mathit{plus}\ \langle x,y \rangle \triangleq \mathbf{ca}&\mathbf{se}\ y\ 
  \mathbf{of} \\[-1\jot]
  & 
  \arraycolsep=2pt
  \begin{array}{lll}
    Z & \to & \begin{array}{l}
      \lfloor \langle x \rangle \rfloor
    \end{array}\\
    S(u) & \to & \begin{array}[t]{l}
      \mathbf{let}\ \langle x',u' \rangle = \mathit{plus}\ \langle x, u 
      \rangle\ \mathbf{in} \\
      \langle x', S(u') \rangle
    \end{array}
  \end{array}
  \\[2\jot]
  \mathit{fib}\ n \triangleq \mathbf{ca}&\mathbf{se}\ n\ \mathbf{of} \\[-1\jot]
  & 
  \arraycolsep=2pt
  \begin{array}{lll}
    Z & \to & \begin{array}{l}
      \langle S(Z), S(Z)\rangle
    \end{array}\\
    S(m) & \to & \begin{array}[t]{l}
      \mathbf{let}\ \langle x,y \rangle = \mathit{fib}\ m\ \mathbf{in} \\
      \mathbf{let}\ z = \mathit{plus}\ \langle y,x \rangle\ \mathbf{in}\ z
    \end{array}
  \end{array}
\end{align*}
\normalsize

%% file: rfun_fib_inv.tex
\footnotesize
\begin{align*}
  \mathit{plus}^{-1}\ z \triangleq \mathbf{ca}&\mathbf{se}\ z\ 
  \mathbf{of} \\[-1\jot]
  & 
  \arraycolsep=2pt
  \begin{array}{lll}
    \lfloor \langle x \rangle \rfloor & \to & \begin{array}{l}
      \langle x, Z \rangle
    \end{array}\\
    \langle x', S(u') \rangle & \to & \begin{array}[t]{l}
      \mathbf{let}\ \langle x,u \rangle = \mathit{plus}^{-1}\ \langle x', u' 
      \rangle\ \mathbf{in} \\
      \langle x, S(u) \rangle
    \end{array}
  \end{array}
  \\[2\jot]
  \mathit{fib}^{-1}\ z \triangleq \mathbf{ca}&\mathbf{se}\ z\ \mathbf{of} \\[-1\jot]
  & 
  \arraycolsep=2pt
  \begin{array}{lll}
    \langle S(Z), S(Z)\rangle & \to & \begin{array}{l}
      Z
    \end{array}\\
    z' & \to & \begin{array}[t]{l}
      \mathbf{let}\ \langle y,x \rangle = \mathit{plus}^{-1}\ z'\ \mathbf{in} 
      \\
      \mathbf{let}\ m = \mathit{fib}^{-1}\ \langle x,y \rangle\ \mathbf{in}\ 
      S(m)
    \end{array}
  \end{array}
\end{align*}
\normalsize